\documentclass{ifacconf}

\usepackage[english]{babel}
\usepackage{amssymb}
\usepackage{stmaryrd}
\usepackage{tikz}
\usepackage{mathtools}
\usepackage{etoolbox}
\usepackage{tabularx}

\usepackage[ruled,vlined,algo2e,linesnumbered]{algorithm2e}
\usepackage{nicematrix}
\usepackage{url}

\DeclareMathAlphabet{\pazocal}{OMS}{zplm}{m}{n}
\usepackage [autostyle, english = american]{csquotes} 
\MakeOuterQuote{"}

\usepackage{xargs}
\usepackage{soul, color, xcolor, graphicx}
\usepackage{enumitem}
\usepackage{natbib}

\usetikzlibrary{positioning,arrows,petri,calc,decorations.markings,arrows.meta,patterns,patterns.meta}
\tikzset{
place/.style={circle,thick,minimum size=4mm,draw},
transitionV/.style={rectangle,thick,fill=black,minimum height=6mm,inner xsep=1pt}
}
\usepackage{pgfplots}
\pgfplotsset{compat=newest}

\definecolor{myblue}{RGB}{0, 101, 202}
\definecolor{mygreen}{RGB}{30, 0, 0}
\definecolor{myred}{RGB}{197, 14, 31}
\definecolor{mypurple}{RGB}{128, 0, 128}
\definecolor{myyellow}{RGB}{204, 204, 0}
\definecolor{mygrey}{RGB}{105, 105, 105}

\definecolor{myblue}{RGB}{0, 101, 202}
\definecolor{mygreen}{RGB}{80, 130, 0}
\definecolor{myred}{RGB}{197, 14, 31}
\definecolor{mypurple}{RGB}{128, 0, 128}
\definecolor{myyellow}{RGB}{204, 204, 0}
\definecolor{mygrey}{RGB}{105, 105, 105}
\colorlet{LightGray}{mygrey!15!white}

\newcommand{\myred}[1]{\textbf{\textcolor{myred}{#1}}}
\newcommand{\myblue}[1]{\textbf{\textcolor{myblue}{#1}}}
\newcommand{\mygreen}[1]{\textbf{\textcolor{mygreen}{#1}}}

\newcommand{\dint}[1]{\left\llbracket#1\right\rrbracket} % discrete interval

\newcommand{\D}{\mathcal{D}}
\renewcommand{\S}{\mathcal{S}}

\newcommand{\N}{\mathbb{N}}
\newcommand{\No}{\mathbb{N}_0}
\newcommand{\Z}{\mathbb{Z}}

\newcommand{\R}{\mathbb{R}}

\newcommand{\Rmax}{{\R}_{\normalfont\fontsize{7pt}{11pt}\selectfont\mbox{max}}}

\newcommand{\Rbar}{\overline{\R}}

\newcommand{\places}{\mathcal{P}}
\newcommand{\transitions}{\mathcal{T}}
\newcommand{\marking}{m}
\newcommand{\arcs}{\mathcal{A}}

\makeatletter
\newcommand{\splus}{%
  \DOTSB\mathop{\mathpalette\mattos@splus\relax}\slimits@
}
\newcommand\mattos@splus[2]{%
  \vcenter{\hbox{%
    \sbox\z@{$#1\oplus$}%
    \resizebox{!}{0.9\dimexpr\ht\z@+\dp\z@}{\raisebox{\depth}{$\m@th#1\boxplus$}}%
  }}%
  \vphantom{\oplus}%
}
\makeatother

\makeatletter
\newcommand{\stimes}{%
  \DOTSB\mathop{\mathpalette\mattos@stimes\relax}\slimits@
}
\newcommand\mattos@stimes[2]{%
  \vcenter{\hbox{%
    \sbox\z@{$#1\otimes$}%
    \resizebox{!}{0.9\dimexpr\ht\z@+\dp\z@}{\raisebox{\depth}{$\m@th#1\boxtimes$}}%
  }}%
  \vphantom{\otimes}%
}
\makeatother

\usepackage{pict2e}

\makeatletter
\newcommand*{\bigsplus}{\DOTSB\mathop{\mathpalette\big@boxplus\relax}\slimits@}

\newcommand{\big@boxplus}[2]{%
  \vcenter{%
    \m@th\bigbox@thickness{#1}%
    \sbox\z@{$#1\bigoplus$}%
    \dimen@=\ht\z@ \advance\dimen@\dp\z@
    \hbox{%
      \setlength{\unitlength}{\dimen@}%
      \begin{picture}(1,1)
      \polyline(0.1,0.1)(0.9,0.1)(0.9,0.9)(0.1,0.9)(0.1,0.1)(0.5,0.1)
      \polyline(0.5,0.1)(0.5,0.9)
      \polyline(0.1,0.5)(0.9,0.5)
      \end{picture}%
    }%
  }%
}

\newcommand{\bigbox@thickness}[1]{%
  \ifx#1\displaystyle
    \linethickness{0.2ex}%
  \else
    \ifx#1\textstyle
      \linethickness{0.16ex}%
    \else
      \ifx#1\scriptstyle
        \linethickness{0.12ex}%
      \else
        \linethickness{0.1ex}%
      \fi
    \fi
  \fi
}
\makeatother

\newcommand{\svdots}{\raisebox{3pt}{$\scalebox{.6}{\vdots}$}} 
\newcommand{\sddots}{\raisebox{3pt}{$\scalebox{.6}{$\ddots$}$}}

\renewcommand{\epsilon}{\varepsilon}

\newcommand\myOverwrite[2]{\!\makebox[0cm][l]{#1}#2\ \!} % for divisions
\newcommand{\ldiv}{\textup{\hspace{0.85mm}\myOverwrite{$\circ$}{$\backslash$}\hspace{-0.5mm}}} % left division
\newcommand{\rdiv}{\textup{\hspace{0.85mm}\myOverwrite{$\circ$}{$\slash$}\hspace{-0.5mm}}} % right division

\newcommand{\eqtop}[1]{\mathrel{\overset{{\mbox{\normalfont\scriptsize #1}}}{=}}}
\newcommand{\geqtop}[1]{\mathrel{\overset{{\mbox{\normalfont\scriptsize #1}}}{\geq}}}

\newcommand{\ie}{i.e., }

\renewcommand{\e}{e}
\newcommand{\epsi}{\varepsilon}

\newcommand{\mathqedhere}{\tag*{$\blacksquare$}} 

\newcommand*{\qedhere}[1][$\blacksquare$]{%
\leavevmode\unskip\penalty9999 \hbox{}\nobreak\hfill
    \quad\hbox{#1}%
}

\begin{document}

% double spacing for Joerg:
%\linespread{2}
%\openup 1em

\begin{frontmatter}

\title{A Luenberger Observer\\for P-Time Event Graphs} 

\author[First]{Dominik Tirpák}
\author[First]{Davide Zorzenon}
\author[First,Second]{J\"{o}rg Raisch\thanksref{footnoteinfo}}

\thanks[footnoteinfo]{
Support from Deutsche Forschungsgemeinschaft (DFG) via grant
RA 516/14-1 and under Germany’s Excellence Strategy – EXC
2002/1 “Science of Intelligence” – project number 390523135 is
gratefully acknowledged.
This paper is derived from the master's thesis of the first author, \cite{tirpak2024development}.
%This work was funded by the Deutsche Forschungsgemeinschaft (DFG, German Research Foundation), Projektnummer RA 516/14-1.
%Partially supported by the Deutsche Forschungsgemeinschaft (DFG, German Research Foundation), under Germany's Excellence Strategy - EXC 2002/1 "Science of Intelligence" - project number 390523135.
}

\address[First]{Control Systems Group, Technische Universit\"{a}t Berlin, Germany (email:~tirpakdominik@gmail.com, [zorzenon,raisch]@control.tu-berlin.de).}
\address[Second]{Science~of~Intelligence,~Research~Cluster~of~Excellence,~Berlin,~Germany.}

\begin{abstract}
P-Time Event Graphs (P-TEGs) are discrete event systems able to model synchronization and delay phenomena.
They extend the modeling power of Timed Event Graphs (TEGs) by including not only lower-bound, but also upper-bound constraints on the sojourn times of tokens in places.
In this work, we consider the problem of estimating the firing time of transitions in P-TEGs, assuming that only a subset of transitions can be directly observed.
Building on the (now classical) Luenberger observer for TEGs, we design an algorithm that takes into account the additional restrictions posed by the upper-bound constraints of P-TEGs to obtain a more accurate firing time estimation.
%A simple example of a P-TEG is used to illustrate the benefits of this approach.
%Our algorithm is optimal, in the sense that the estimated firing times are as large as possible (according to the available information) but cannot be later than the actual firing times.
\end{abstract}


\end{frontmatter}
%===============================================================================

\section{Introduction}

P-Time Event Graphs (P-TEGs) are mathematical tools for modeling time-critical man-made systems, which are ubiquitous in, e.g, the manufacturing, chemical, and food industry (\cite{khansa1996p}).
They are able to represent synchronization of events (e.g., both tasks A and B need to end before starting task C) and time-window constraints (e.g., deadlines, under-/over-processing constraints).

In this paper, we adopt the point of view of an external observer of a system described by a P-TEG.
The observer can only measure the time when certain events occur (e.g., input material enters the plant, or a finished product is shipped), while other events remain hidden (e.g., a semi-finished product enters a certain processing stage) and may be subject to unexpected delays.
The aim of the observer is to reconstruct, in the best possible way and during the operation of the system (i.e., online), the time in which events of the latter type occur.

To solve this problem, we extend the Luenberger observer for timed event graphs (TEGs) introduced in \cite{hardouin2010observer} and discussed in \cite[Chapter 10]{hardouin2018control} to the case of P-TEGs.
Unlike P-TEGs, TEGs are unable to model temporal upper-bound constraints, which our observer can exploit through an internal model of the P-TEG to get more refined estimations of the system behavior.
On the technical side, this extension is made possible by the adoption of a more flexible alternative to the well-known $\gamma$-transform, called infinite vector representation. %which is used to conveniently transform recurrence equations and inequalities into algebraic ones.
Both of these techniques are used to convert linear max-plus difference equations and inequalities into algebraic ones.
The difference is that, unlike the $\gamma$-transform, representing the dynamics by means of infinite vectors allows to deal with several types of initial conditions, switching, and the sort of inequalities characterizing P-TEGs---for a longer discussion on the issues related to the application of the $\gamma$-transform in P-TEGs, we refer to \cite[Section 7.3]{zorzenon2025modeling}.

\subsection{Related work}
%Our problem can be considered as a subproblem of 
Besides \cite{hardouin2010observer}, the state estimation problem for TEGs has been considered in several works, see e.g.,  \cite{Espindola2022stochastic,di2010duality} and references therein; see also \cite{TRUNK2020501}, where the Luenberger observer from~\cite{hardouin2010observer} is generalized to weighted timed event graphs.

Our problem is related to the marking estimation problem for P-time Petri nets, which was addressed in \cite{bonhomme2015marking}. % and extended to a decentralized framework in \cite{Bonhomme2021decentralized}.
Unlike P-TEGs, P-time Petri nets are additionally able to model conflicts.
The main difference between our approaches is that our observer provides the best occurrence time estimates, i.e., the greatest under-approximation of the occurrence times of events, whereas the algorithm by Bonhomme computes all possible sequences of events that are consistent with the available information.
An advantage of our method is that it does not require the reachability set of the underlying Petri net to be finite.
%Moreover, the extension of our observer to larger classes of systems 
%, and \cite{Elpindola2022stochastic}, where the estimation problem for TEGs in presence of model uncertainties is addressed.

\textbf{Notation.}
Symbols $\N$, $\No$, $\Z$, $\R_{\geq 0}$, $\R$ denote the set of positive integers, nonnegative integers, integers, nonnegative reals, and reals, respectively.
Moreover, $\Rmax\coloneqq\R\cup\{-\infty\}$ and $\Rbar\coloneqq\R\cup\{-\infty,+\infty\}$.
$A\in \D^{m\times n}$ denotes a matrix with dimensions $m\times n$ and elements from set $\D$, and $A^T$ denotes the transpose of matrix $A$.
For a column vector of length $p$, the notation $b\in \D^{p}$ is used.
% Discrete interval
Given $a,b\in\Z$ with $b \geq a$, $\dint{a, b}$ indicates the discrete interval $[a, b]\cap \Z$ = $\{ a, \dots, b \}$.
\section{Algebraic preliminaries}

This section briefly recalls the necessary algebraic tools.
For a more detailed overview of the topics covered here, we refer to \cite{baccelli1992synchronization,hardouin2018control}.

% Dioids in general
\subsection{Idempotent semirings}
\label{subsec: dioids}

To model the dynamics of P-TEGs, the max-plus algebra is employed, which is a complete idempotent semiring (or complete dioid).
A dioid is a set with the binary operations $\oplus$ (dioid addition) and $\otimes$ (dioid multiplication, symbol often omitted).
$\oplus$ is associative, commutative and idempotent, with neutral (zero) element $\varepsilon$.
$\otimes$ distributes over $\oplus$, is associative and has neutral (identity) element $e$.
Also, the zero element $\varepsilon$ is absorbing for dioid multiplication.
A dioid is complete if it is closed for infinite sums and and if multiplication distributes over infinite sums.

The order relation $\leq$ of a dioid is defined using addition by $a\leq b\ \Leftrightarrow \ a\oplus b = b$ for all $a, b$ from the dioid.
Operations $\oplus$, $\otimes$ are order preserving, meaning that $a\leq b \ \Leftrightarrow\ a\oplus c\leq b\oplus c$ and $ac\leq bc$.
Moreover, $\oplus$ satisfies: $a\leq b$ and $c\leq b$ $\Leftrightarrow$ $a\oplus c\leq b$.
Another binary operation, the infimum (also associative, commutative and idempotent) can be defined: $a\leq b\ \Leftrightarrow\ a\wedge b = a$.
The Kleene star $^*$ and the Kleene plus $^+$ operators are defined by $a^* = e\oplus a\oplus a^2\oplus a^3\oplus \dots$, with $a^0=e$ and $a^{k+1} = a\otimes a^k$ for all $k\in\N_0$, and $a^+ = a a^*$.
% Left and right divison
From the so-called residuation theory, we obtain two additional important operations: left division $\ldiv$ and right division $\rdiv$.
$a\ldiv b$ (read "$b$ left divided by $a$") can be defined as the unique element for which the equivalence $x\leq a\ldiv b \ \Leftrightarrow\ ax\leq b$ is satisfied; similarly, the result of $c\rdiv d$ (read "$c$ right divided by $d$") is defined such that $x\leq c\rdiv d \ \Leftrightarrow\ xd\leq c$.
Moreover, we denote $a\rdiv a$ by $a^\times$.
%Residuation theory guarantees that operations $\ldiv$ and $\rdiv$ are well-defined.
%provides the greatest solution $x$ of inequality $a x\leq b$, whereas $c\rdiv d$ will be the greatest solution to $x d\leq c$.

% TODO: not sure if all of these properties are important in this paper
%\renewcommand\arraystretch{0}
%\renewcommand\tabularxcolumn[1]{m{#1}}% for vertical centering text in X column
Properties of the operations defined so far are listed below:\\[-.2cm]
\begin{subequations}
  \begin{tabularx}{\linewidth}{p{3cm}X}
  \begin{equation}\label{eq:a**}
    (a^*)^*=a^*,
  \end{equation}
  & 
  \begin{equation}\label{eq:a*a*}
    a^*a^*=a^* ,
    \end{equation}
    \end{tabularx}
    \end{subequations}\\[-.8cm]
\begin{subequations}
  \begin{tabularx}{\linewidth}{p{3cm}X}
  \begin{equation}\label{eq:(x/a)a}
    (x\rdiv a) a \leq x,
  \end{equation}
  & 
  \begin{equation}\label{eq:x/a+b}
   x\rdiv(a\oplus b) = x\rdiv a \wedge x\rdiv b.
  \end{equation}
  \end{tabularx}
\end{subequations}\\[-.4cm]
%\begin{align}
%     &  \\
%    & 
%\end{align}
%\begin{equation}
%\end{equation}
%$a\rdiv a = (a\rdiv a)^*$ and $(a\rdiv a) a = a$.
%Note that $a^* \geq e$ holds for any dioid element $a$. Also, $a^* a^* = a^*$ and $(a^*)^* = a^*$. The Kleene plus operator can then be defined by $a^+ = aa^*$.
%The following 
\begin{thm}\label{th:star}
    The least solution of $x = ax\oplus b$ is $x = a^* b$.
\end{thm}
%The Kleene star theorem states that in a complete dioid, 

% The max-plus algebra
\subsection{Max-plus algebra}

The max-plus algebra is a complete dioid, containing elements from $\Rbar = \R \cup \{ +\infty,-\infty \}$.
The operation $\oplus$ corresponds to the standard $\mathrm{\max}$ operation, $\otimes$ to standard addition (with the additional rule $-\infty + (+\infty)=+\infty + (-\infty)=-\infty$), while the infimum $\wedge$ is the standard $\mathrm{min}$ operation.
Its order relation corresponds to the standard order $\leq$.
The zero element is $\varepsilon = -\infty$, and the identity element is $e = 0$.
Left and right divisions correspond to standard subtraction: $a\ldiv b=b-a$ and $a\rdiv b = a-b$ (with the additional rule $+\infty-(+\infty)=-\infty-(-\infty)=+\infty$).

% Matrix dioid
For any sets $\D$ and $M,N\subseteq \Z$, a function $A:M\times N\rightarrow \D$ is called matrix; we denote $A(i,j)$ by $A_{ij}$ and the set of all $M\times N$ matrices with codomain $\D$ by $\D^{M\times N}$.
The set of square matrices of finite or infinite dimension $\Rmax^{N\times N}$, where $N\subseteq\Z$, with elements from the max-plus algebra, is also a complete dioid.
The extension of the dioid operations to matrices works in an analogous way to the standard algebra.
In particular, $(A\oplus B)_{ij} = A_{ij} \oplus B_{ij}$ and $(A\otimes B)_{ij} = \bigoplus_k A_{ik}\otimes B_{kj}$.
Moreover, $(A\wedge B)_{ij} = A_{ij}\wedge B_{ij}$ and $(A\rdiv B)_{ij} = \bigwedge_k A_{ik}\rdiv B_{jk}$.
%The same holds for multiplication of a matrix by a scalar.
Here, the zero matrix is the matrix $\mathcal{E}$ with elements all equal to $\varepsilon$.
The identity matrix $E$ is a square matrix with $e$ on its main diagonal and $\varepsilon$ elsewhere.
Note that for matrices, $\leq$ is a partial order, defined by addition as described in Section \ref{subsec: dioids}.
%In the reminder of the paper, vector $v$ of length $q$ containing only $e$, resp. $\varepsilon$ elements, will be denoted $e_q$, resp. $\varepsilon_q$.

% Infinite matrix dioid
%An interesting extension of matrix dioids is the case of infinitely large matrices. The set of matrices of infinite size, here denoted by $\Rmax^{\N_0 \times \N_0}$, also forms an infinite dioid.

%To compute $\rdiv$ for matrices, consider $A\in\Rmax^{n\times m}$ and $Y\in\Rmax^{n\times p}$, as well as $Z\in\Rmax^{p\times m}$. Then,
%%$$(A\ldiv Y)_{ij} = \bigwedge_{k=1}^n (a_{ki}\ldiv y_{kj}),\ i = 1,\, \dots,\, m\ \mathrm{and}\ j = 1,\, \dots,\, p\, ;$$
%$$(Z\rdiv A)_{ij} = \bigwedge_{k=1}^m (z_{ik}\rdiv a_{jk}),\ i = 1,\, \dots,\, p\ \mathrm{and}\ j = 1,\, \dots,\, n\, .$$

\section{P-time event graphs}

In this section, we recall the definition of P-TEGs and their dynamics.

\subsection{General definition}\label{su:definition}

\begin{defn}[\cite{khansa1996p}]\label{de:PTEG}
A P-TEG is a 5-tuple $(\places,\transitions,\arcs,\marking,\iota)$, in which:
\begin{itemize}
    \item $\places$ is a finite set of places,
    \item $\transitions$ is a finite set of transitions,
    \item $\arcs\subseteq (\places\times \transitions)\cup (\transitions \times \places)$ is the set of arcs connecting places to transitions and transitions to places,
    \item $\marking:\places\rightarrow\No$ and $\iota:\places\rightarrow\{[\tau^-,\tau^+]\cap \R \mid \tau^-\in \R_{\geq 0},\tau^+\in\R_{\geq 0}\cup\{+ \infty\}\}$ are two maps that associate to each place $p\in\places$, respectively, its initial number of tokens (or markings) $\marking(p)$, and a time interval $\iota(p)=[\tau_p^-,\tau_p^+]\cap \R$,
    \item every place has exactly one upstream transition and one downstream transition, \ie $\forall p\in\places$, $\exists ! t_\textup{u},t_\textup{d}\in\transitions$ such that $(t_\textup{u},p),(p,t_\textup{d})\in\arcs$.
\end{itemize}
\end{defn}

A transition $t\in\transitions$ is said to be enabled if either it has no upstream places (i.e., $\forall p\in\places$, $(p,t)\notin\arcs$) or each upstream place $p\in \places$ contains at least one token that has resided in $p$ for a time included in interval $[\tau_p^-,\tau_p^+]\cap\R$.
Note that this time interval is always closed, unless $\tau_p^+=+\infty$, in which case it is of the form $[\tau_p^-,+\infty)$.
When transition $t$ is enabled, it can fire, causing one token to be instantaneously removed from each upstream place and one token to be instantaneously added to each downstream place.
More precise conditions for the firing of transitions will be discussed in Section~\ref{su:assumptions}.

If a token resides for too long in $p$, violating the constraint imposed by interval $[\tau_p^-,\tau_p^+]\cap \R$, then the token is said to be \textit{dead}.
In this article, we will only consider \emph{consistent} trajectories of P-TEGs in which no token death occurs.
Recall that a timed event graph (TEG) is a P-TEG in which $\tau_p^+=+\infty$ for all $p\in\places$.

\subsection{Model assumptions}\label{su:assumptions}

In this work, we categorize the transitions of P-TEGs as follows:
\begin{itemize}
    \item internal transitions $x_i$, $i\in\dint{1,n_x}$, which possess both upstream and downstream places,
    \item output transitions $y_i$, $i\in\dint{1,n_y}$, having only upstream places,
    \item input transitions, with downstream places only, which can be further distinguished as:
    \begin{itemize}
        \item observable input transitions $u_i$, $i\in\dint{1,n_u}$,
        \item unobservable input transitions $w_i$, $i\in\dint{1,n_w}$.
    \end{itemize}
\end{itemize}
Moreover, we assume that internal and output transitions (but not input transitions) follow the \emph{earliest firing rule}, \ie they fire as soon as they are enabled. %; input transitions on the other hand are not required to fire the moment in which they are enabled.

The distinction between observable and unobservable input transitions indicates which information is available to the observer that will be developed in the next section.
In particular, only the firing times of output transitions and observable input transitions can be accessed directly by the observer, while the firing times of internal transitions and unobservable input transitions cannot be directly measured.
The goal of the observer will precisely be to estimate the firing times of internal transitions.

The following structural assumptions are made solely to simplify the subsequent discussions, and do not limit the generality of the results\footnote{In the sense that any P-TEG can be transformed into a behaviorally equivalent one that satisfies the given assumption, at the cost of increasing the number of transitions and places.}:
\begin{enumerate}[label=A\arabic*]
    \item\label{en:assumptiontokens} each place $p\in\places$ has either zero or one initial token, \ie $m(p)\in\{0,1\}$,
    \item\label{en:assumptioninputs} each input (resp., output) transition has exactly one downstream (resp., upstream) place $p$, which has zero initial tokens, $m(p)=0$, and time interval $\iota(p)=[0,+\infty)$, and the downstream (resp., upstream) transition of $p$ is internal,
    %\item\label{en:assumption3} each input transition (either observable or unobservable) has exactly one downstream place $p$ with zero initial tokens, $m(p)=0$, time interval $\iota(p)=[0,+\infty)$, and the downstream transition of $p$ is internal,
    %\item[A3]\label{en:assumption4} 
    \item\label{en:assumptioninitial} initial tokens "appear" in their respective places at time $0$. %, and they are allowed to stay in these places for an arbitrary amount of time (even larger than what the upper-bound constraint would specify)\footnote{Other types of initial conditions have been considered in the literature (see, e.g., \cite{zorzenon2023switched}).}.
\end{enumerate}
Finally, our last assumption is going to simplify the computations in Section \ref{se:observer}:
\begin{enumerate}[label=A\arabic*]\setcounter{enumi}{3}
    \item\label{en:assumptionuncontrollable} for each internal transition $x_i$, there is exactly one unobservable input transition $w_i$ whose downstream place is upstream of $x_i$.
\end{enumerate}
Assumption \ref{en:assumptionuncontrollable} is indeed conservative, and generalizing our results to P-TEGs in which it does not hold is left for future work.

\begin{exmp}
\begin{figure}[t]
    \centering
    %\resizebox{.8\linewidth}{!}{
    \begin{tikzpicture}[node distance=1cm and 1.2cm,bend angle=30,thick,on grid]
\footnotesize

\node [transitionV,label=below:{$x_1$}] (x1) {};
\node [place,above=1cm of x1,tokens=1,label=above:{$[2,10]$}] (p11) {};
\node [place,below right=.7cm and 1.2cm of x1,label=above:{$[0,5]$}] (p12) {};
\node [transitionV,right=of p12,label=below:{$x_2$}] (x2) {};
\node [place,right=of x2,label=above:{$[0,0]$}] (p23) {};
\node [transitionV,above right=.7cm and 1.2cm of p23,label=below:{$x_3$}] (x3) {};
\node [place,myblue,right=of x3] (p3y) {};
\node [transitionV,right=of p3y,myblue,label=below:{\myblue{$y_1$}}] (y) {};
\node [place,above=.7cm of x2,label=above:{$[1,+\infty)$}] (p13) {};
\node [transitionV,above=of p13,myred,label=above:{\myred{$w_1$}}] (w1) {};
\node [transitionV,above=of y,myred,label=above:{\myred{$w_3$}}] (w3) {};
\node [place,myred,left=of w3] (pw3) {};
\node [place,myred,left=of w1] (pw1) {};
\node [place,myred,below=.7cm of p23] (pw2) {};
\node [transitionV,right=of pw2,myred,label=above:{\myred{$w_2$}}] (w2) {};
\node [place,mygreen,below=.7cm of p12] (pu2) {};
\node [transitionV,left=of pu2,mygreen,label=above:{\mygreen{$u_1$}}] (u) {};

\draw [-stealth'] (x1.-90+20) to (p12);
\draw [-stealth'] (x1.90-10) [bend right] to (p11);
\draw [-stealth'] (p11) [bend right] to (x1.90+10);
\draw [-stealth'] (x1) to (p13);
\draw [-stealth',myred] (pw1) to (x1.90-20);
\draw [-stealth',myred] (w1) to (pw1);
\draw [-stealth',myred] (pw3) to (x3.90-20);
\draw [-stealth',myred] (w3) to (pw3);
\draw [-stealth'] (p12) to (x2);
\draw [-stealth'] (x2) to (p23);
\draw [-stealth'] (p23) to (x3.-90-20);
\draw [-stealth'] (p13) to (x3);
\draw [-stealth',myblue] (p3y) to (y);
\draw [-stealth',myblue] (x3) to (p3y);
\draw [-stealth',mygreen] (pu2) to (x2.-90-20);
\draw [-stealth',mygreen] (u) to (pu2);
\draw [-stealth',myred] (pw2) to (x2.-90+20);
\draw [-stealth',myred] (w2) to (pw2);

\end{tikzpicture}
    %}
    \caption{Example of P-TEG. Where not indicated, the interval associated to places is $[0,+\infty)$.}\label{fi:P-TEG_example}
\end{figure}
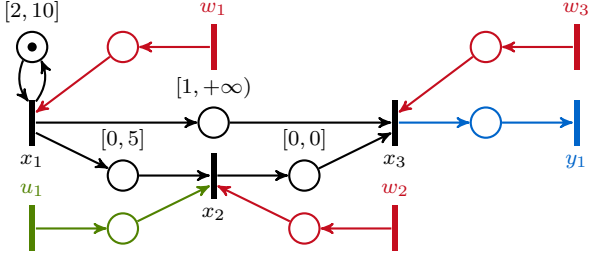
    Figure \ref{fi:P-TEG_example} shows an example of a P-TEG.
    Transitions are represented by bars, places by circles, initial tokens by dots, and arcs by arrows; the time interval associated to places is shown in the proximity of the corresponding place.
    The P-TEG has internal transitions $x_1,\,x_2,\,x_3$, observable input transition $u_1$, unobservable input transitions $w_1,\,w_2,\,w_3$, and output transition $y_1$. %\qedhere
    %According to the assumptions above, the first firing time of transition $x_1$ must be at least at time $2$, because of the lower bound on the place upstream and downstream of $x_1$, but can be 
\end{exmp}

\subsection{Dynamics of P-TEGs}

%The behavior of P-TEGs can be represented by equations and inequalities 
In order to represent the behavior of P-TEGs by means of equations and inequalities, it is convenient to introduce dater functions.
Let us define dater functions $x:\N_0\rightarrow\R^{n_x}$, $y:\N\rightarrow\R^{n_y}$, $u:\N\rightarrow\R^{n_u}$, and $w:\N\rightarrow\R^{n_w}$ so that, for all $k\geq 1$, $x_i(k)$, respectively $y_i(k)$, $u_i(k)$, and $w_i(k)$, represents the time when transition $x_i$, respectively $y_i$, $u_i$, $w_i$, fires for the $k$-th time. 
Since the $(k+1)$-st firing of each transition cannot occur before the $k$-th one, the dater functions are nondecreasing. 
For handling initial conditions, we set the auxiliary variable $x_i(0)$ to $0$ for all $i$; in this way, all components of $x(0)$ assume the meaning of "initial time $t=0$".

\begin{rem}\label{re:unrolling}
    Note that dater function $x:\N_0\rightarrow\R^{n_x}$ can be interpreted as an infinite-dimensional vector $x\in\R^{\N_0}$ in which entry $x_i$ corresponds to element $x_j(k)$, where 
    \begin{equation}\label{eq:k_and_j}
        k = \left\lceil \frac{i}{n_x} \right\rceil-1 \mbox{ and } j = (i-1 \mod n_x)+1.
    \end{equation}
    Similar interpretations for $y,u,w$ hold; for instance, $y:\N\rightarrow\R^{n_y}$ is associated with vector $y\in\R^{\N}$ such that $y_i = y_j(k)$, where $k = \left\lceil \frac{i}{n_y} \right\rceil$ and $j = (i-1\mod n_y)+1$.
    This point of view will be adopted starting from Section~\ref{su:unrolling}.%\qedhere
\end{rem}

Recall (from, e.g., \cite{hardouin2018control}) that dater functions of a P-TEG in which internal and output transitions follow the earliest firing rule must satisfy for all $k\geq 1$:
\begin{align}
    \label{eq:x0} x(0) &= e_{n_x},\\
    \label{eq:EFR} x(k) &= A_0 x(k)\oplus A_1x(k-1)\oplus B_0 u(k)\oplus R_0 w(k),\\
    \label{eq:output} y(k)&=C_0x(k),
\end{align}
where $e_{n_x}$ is a vector of $n_x$ $e$'s, $A_0,A_1\in\Rmax^{n_x\times n_x}$ are defined so that, given $\mu\in\{0,1\}$ and $i,j\in\dint{1,n_x}$, $(A_{\mu})_{ij}=\tau_p^{-}$ if there is a place $p$ with $m(p)=\mu$ initial tokens (see \ref{en:assumptiontokens}), upstream transition $x_j$, downstream transition $x_i$, and time interval $\iota(p)=[\tau_p^-,\tau_p^+]\cap\R$, and $(A_{\mu})_{ij}=\epsi$ otherwise; similarly, $B_0\in\{\epsi,\e\}^{n_x\times n_u}$, $R_0\in\{\epsi,\e\}^{n_x\times n_w}$, $C_0\in\{\epsi,\e\}^{n_y\times n_x}$ are defined so that elements $(B_0)_{ij}$, respectively $(R_0)_{ij}$, $(C_0)_{ij}$ are $\e$ if there is a place (with $0$ initial tokens and time interval $[0,+\infty)$, by \ref{en:assumptioninputs}) connecting transitions $u_j$ and $x_i$, respectively $w_j$ and $x_i$, $x_j$ and $y_i$, and are $\epsi$ otherwise.
Due to \ref{en:assumptionuncontrollable}, $R_0$ is the identity matrix $E$ and therefore $n_w = n_x$.
We remark that \eqref{eq:x0} is used to handle the initial conditions according to \ref{en:assumptioninitial}.
%; note that this equation implies:
%\begin{equation}
%    \label{eq:x0_inequality}x(0)\geq \mathbb{E}x(0),
%\end{equation}
%where $\mathbb{E}_{ij}=\e$ for all $i,j\in\dint{1,n_x}$.

The dynamical equations presented so far ignore the upper-bound constraints on the residence time of tokens in places (thus, they are adequate to model TEGs but not general P-TEGs).
In order to take them into account, we need to impose additional inequalities on dater $x$.
Consider a place $p$ with $m(p)\in\{0,1\}$ initial tokens, upstream transition $x_j$, downstream transition $x_i$, and time interval $\iota(p)=[\tau_p^-,\tau_p^+]\cap\R$.
Then, the inequality $x_j(k-m(p)) + \tau_p^+ \geq x_i(k)$ must be satisfied for all $k\geq 1$.
Equivalently, this inequality can be written as $x_j(k-m(p)) \geq -\tau_p^+ + x_i(k)$.
We can group together these inequalities for all internal transitions into two max-plus vector inequalities: for all $k\geq 1$,
\begin{align}
    \label{eq:U0}x(k) &\geq U_0 x(k),\\
    \label{eq:U1}x(k-1) &\geq U_1 x(k),
\end{align}
where $U_0,U_1\in\Rmax^{n\times n}$ are matrices defined such that $(U_{\mu})_{ji} = -\tau_p^+$ if there is a place $p$ with $m(p)=\mu\in\{0,1\}$ initial tokens, upstream transition $x_j$, downstream transition $x_i$, and time interval $\iota(p)=[\tau_p^-,\tau_p^+]\cap\R$, and $(U_\mu)_{ji}=\epsi$ otherwise.

Moreover, to model the nondecreasingness of dater function $x$, it is necessary to impose for all $k\geq 1$:
\begin{align}
    \label{eq:nondecreasing}x(k)&\geq Ex(k-1).
\end{align}

\begin{exmp}
Consider again the P-TEG in Figure~\ref{fi:P-TEG_example}.
According to Sections \ref{su:definition} and \ref{su:assumptions}, the dater functions need to satisfy \eqref{eq:x0}-\eqref{eq:nondecreasing}, with 
\[
    A_0=
    \begin{bmatrix}
        \epsi & \epsi & \epsi\\
        \e & \epsi & \epsi\\
        1 & \e & \epsi
    \end{bmatrix},\
    A_1=
    \begin{bmatrix}
        2 & \epsi & \epsi\\
        \epsi & \epsi & \epsi\\
        \epsi & \epsi & \epsi
    \end{bmatrix},\
    B_0=
    \begin{bmatrix}
        \epsi \\
        \e \\
        \epsi 
    \end{bmatrix},\
    R_0=
    \begin{bmatrix}
        \e & \epsi & \epsi\\
        \epsi & \e & \epsi\\
        \epsi & \epsi & \e
    \end{bmatrix},
\]
\[
    C_0 = 
    \begin{bmatrix}
        \epsi&\epsi&\e
    \end{bmatrix},\
    U_0 = 
    \begin{bmatrix}
        \epsi & -5 & \epsi\\
        \epsi & \epsi & \e\\
        \epsi& \epsi& \epsi
    \end{bmatrix},\
    U_1 = 
    \begin{bmatrix}
        -10&\epsi&\epsi\\
        \epsi&\epsi&\epsi\\
        \epsi&\epsi&\epsi\\
    \end{bmatrix}.
\]

Suppose that all input transitions, i.e., $u_1$ and $w_i$, $i\in\dint{1,3}$, fire for the fist time at time $2$.
Then, $u_1(1) = w_i(1) =  2$, and the unique solution $x(1)$ of \eqref{eq:EFR} for $k=1$ is $x(1) = [2\ 2\ 3]^T$.
However, note that these values of $x(1)$ violate \eqref{eq:U0}, which imposes that $x_2(1) \geq x_3(1)$, in accordance with the upper-bound constraint on the downstream place of transition $x_2$.
Therefore, this choice of input firing times leads to a non-consistent trajectory.

Let us instead consider the following firing times: % of input transitions:
\begin{equation}\label{eq:wexample}
    w(1)\!=\! \begin{bmatrix}
        e\\4\\e
    \end{bmatrix}\!\!,\,
    w(2) \!=\! \begin{bmatrix}
        7\\10\\5
    \end{bmatrix}\!\!,\,
    w(3) \!=\! \begin{bmatrix}
        16\\20\\20
    \end{bmatrix}\!\!,\,
    w(4)\!=\!\begin{bmatrix}
        26\\31\\20
    \end{bmatrix}\!\!,
\end{equation}
and $u_1(k)=0$ for all $k$.
With these values, \eqref{eq:EFR} results in the trajectory
\begin{equation}\label{eq:xexample}
    x(1) = \begin{bmatrix}
        2\\4\\4
    \end{bmatrix}\!\!,\,
    x(2)=\begin{bmatrix}
        7\\10\\10
    \end{bmatrix}\!\!,\,
    x(3)=\begin{bmatrix}
        16\\20\\20
    \end{bmatrix}\!\!,\,
    x(4)=\begin{bmatrix}
        26\\31\\31
    \end{bmatrix}\!\!,
\end{equation}
which satisfies constraints \eqref{eq:U0} for all $k\in\dint{1,4}$ and \eqref{eq:U1} for all $k\in\dint{1,3}$.
Finally, note that $y_1(k) = x_3(k)$ for all $k$, according to \eqref{eq:output}.
%\qedhere
\end{exmp}
%\zor{Comments on fastest trajectory being not consistent? Formulate an input firing sequence such that the observer will exhibit different estimations when considering and not considering upper bound constraints} 

%\subsection{Manipulation of the dynamics}
%
%%For convenience, it is useful to transform \eqref{eq:EFR} in explicit form using Proposition \ref{pr:implicit_explicit}:
%%\begin{equation}
%%    \label{eq:EFR_explicit} x(k) = Ax(k-1)\oplus Bu(k)\oplus Rw(k),
%%\end{equation}
%%where $A=A_0^*A_1$, $B = A_0^*B_0$, and $R = A_0^*R_0$.
%
%In order to obtain it, we start by noting that \eqref{eq:EFR}, \eqref{eq:U0}, \eqref{eq:nondecreasing}, and \eqref{eq:EFR_explicit} imply: for all $k\geq 1$,
%\begin{align}
%    \label{eq:centerbound} x(k)&\geq Q x(k),\\
%    \label{eq:rightbound} x(k)&\geq P x(k-1),
%\end{align}
%where $Q = A_0\oplus U_0$ and $P = A\oplus E$.
%Moreover, \eqref{eq:x0} implies
%
\subsection{Infinite vector representation}\label{su:unrolling}

%In this subsection, we introduce a simple technique, called \emph{unrolling}, to obtain algebraic equations and inequalities from the dynamical system \eqref{eq:U1}, \eqref{eq:EFR_explicit}-\eqref{eq:x0_inequality}.
%Unrolling can be thought of as a generalization of the $\gamma$-transform.
%
%Let us define 
As discussed in Remark~\ref{re:unrolling}, there is a 1-to-1 correspondence between each dater function $x, u, w, y$ and its related infinite-dimensional vector, which we denote by the same symbol.
Adopting this point of view, we can rewrite \eqref{eq:x0}-\eqref{eq:nondecreasing} as a system of two algebraic equations and one inequality over vectors and matrices of infinite dimension:
\begin{align}
    \label{eq:EFR_infinite}x &= x_0 \oplus A x \oplus B u \oplus Rw,\\
    \label{eq:output_infinite}y &= C x,\\
    \label{eq:upper_infinite}x &\geq U x,
\end{align}
where \[
x_0 = \begin{bsmallmatrix}
    e_{n_x}\\\epsilon_{n_x}\\\epsilon_{n_x}\\\epsilon_{n_x}\\\svdots
    \end{bsmallmatrix}\!,\
    A = \begin{bsmallmatrix}
        \mathcal{E} & \mathcal{E} &\mathcal{E} &\mathcal{E} &\cdots\\
        A_1 & A_0 &\mathcal{E} &\mathcal{E} &\cdots\\
        \mathcal{E} & A_1 &A_0 &\mathcal{E} &\cdots\\
        \mathcal{E} & \mathcal{E} &A_1 &A_0 &\cdots\\
        \svdots & \svdots &\svdots &\svdots &\sddots
    \end{bsmallmatrix}\!,\
\]\[
    B = 
\begin{bsmallmatrix}
        \mathcal{E} & \mathcal{E} &\mathcal{E} &\mathcal{E} &\cdots\\
        B_0 & \mathcal{E} &\mathcal{E} &\mathcal{E} &\cdots\\
        \mathcal{E} & B_0 &\mathcal{E} &\mathcal{E} &\cdots\\
        \mathcal{E} & \mathcal{E} &B_0 &\mathcal{E} &\cdots\\
        \svdots & \svdots &\svdots &\svdots &\sddots
    \end{bsmallmatrix}\!,\
    R = 
\begin{bsmallmatrix}
        \mathcal{E} & \mathcal{E} &\mathcal{E} &\mathcal{E} &\cdots\\
        R_0 & \mathcal{E} &\mathcal{E} &\mathcal{E} &\cdots\\
        \mathcal{E} & R_0 &\mathcal{E} &\mathcal{E} &\cdots\\
        \mathcal{E} & \mathcal{E} &R_0 &\mathcal{E} &\cdots\\
        \svdots & \svdots &\svdots &\svdots &\sddots
    \end{bsmallmatrix},\
\]\[
    C = 
    \begin{bsmallmatrix}
        \mathcal{E} & C_0 &\mathcal{E} &\mathcal{E} &\cdots\\
        \mathcal{E} & \mathcal{E} &C_0 &\mathcal{E} &\cdots\\
        \mathcal{E} & \mathcal{E} &\mathcal{E} &C_0 &\cdots\\
        \mathcal{E} & \mathcal{E} &\mathcal{E} &\mathcal{E} &\cdots\\
        \svdots & \svdots &\svdots &\svdots &\sddots
    \end{bsmallmatrix},\ %\]\[
    U = 
    \begin{bsmallmatrix}
        \mathcal{E} & U_1 &\mathcal{E} &\mathcal{E} &\cdots\\
        E & U_0 & U_1 &\mathcal{E} &\cdots\\
        \mathcal{E} & E &U_0 & U_1 &\cdots\\
        \mathcal{E} & \mathcal{E} &E &U_0 &\cdots\\
        \svdots & \svdots &\svdots &\svdots &\sddots
    \end{bsmallmatrix},
\]
and $\epsilon_{n_x}$ indicates a vector of $n_x$ $\epsilon$'s.

%We can now rewrite \eqref{eq:U1}, \eqref{eq:centerbound}-\eqref{eq:x0_inequality} as a single inequality pertaining vectors and matrices with infinitely many entries:
%\begin{equation}\label{eq:Mbound}
%    \underbrace{
%    \begin{bmatrix}
%        x(0)\\x(1)\\x(2)\\x(3)\\x(4)\\\vdots
%    \end{bmatrix}
%}_{x}
%    \geq
%    \underbrace{
%    \begin{bmatrix}
%        \mathbb{E} & \mathcal{E} & \mathcal{E} & \mathcal{E} & \mathcal{E}&\cdots\\
%        P & Q & U_1 & \mathcal{E} &\mathcal{E}& \cdots\\
%        \mathcal{E} & P & Q & U_1 &\mathcal{E}& \cdots\\
%        \mathcal{E} & \mathcal{E} & P & Q &U_1& \cdots\\
%        \mathcal{E} & \mathcal{E} & \mathcal{E}& P & Q & \cdots\\
%        \vdots&\vdots&\vdots&\vdots&\vdots&\ddots
%    \end{bmatrix}
%}_M
%\underbrace{
%    \begin{bmatrix}
%        x(0)\\x(1)\\x(2)\\x(3)\\x(4)\\\vdots
%    \end{bmatrix}
%}_x
%.
%\end{equation}
%We call this technique \emph{unrolling}.
%
%
%Let us use the following notation to indicate some of the blocks of matrix $M^*$:
%\[
%M^* = 
%\begin{bmatrix}
%    \times & \times & \times & \times & \times & \cdots\\
%    M_1 & \times & \times & \times & \times & \cdots\\
%    \times & M_2 & \times & \times & \times & \cdots\\
%    \times & \times & M_3 & \times & \times & \cdots\\
%    \times & \times & \times & M_4 & \times & \cdots\\
%    \vdots & \vdots & \vdots & \vdots & \vdots & \ddots
%\end{bmatrix},
%\]
%where each "$\times$" denotes an $n_x\times n_x$ block of $M^*$ that is not useful for our discussion, and each $M_k$ is an $n_x\times n_x$ matrix.

\subsection{A larger state matrix}

%Note that \eqref{eq:EFR_infinite} gives a necessary condition for a consistent trajectory, which becomes also sufficient when combined with \eqref{eq:upper_infinite}.
%The observer that will be constructed in the following section is based on a new equation for $x$ that results in a stricter necessary condition for consistency compared to \eqref{eq:EFR_infinite}.
%Unlike \eqref{eq:EFR_infinite}, this new equation incorporates the information coming from the upper-bound constraints.
%This new equation incorporates the information coming from the upper-bound constraints.
%a stricter necessary condition for $x$ \eqref{eq:EFR_infinite}, as it incorporates the information coming from the upper-bound constraints.
%The following proposition gives a stricter dynamical equation for the evolution of the dater $x$ in P-TEGs.
The observer that will be constructed in the following section will look for the greatest sub-approximation of $x$.
To achieve this, it is useful to first obtain the largest matrix $M^*$ such that \eqref{eq:EFR_infinite} and \eqref{eq:upper_infinite} imply
\begin{equation}\label{eq:dynamics_M}
    x = x_0 \oplus M^*x \oplus Bu \oplus Rw.
\end{equation}

\begin{prop}\label{pr:dynamics_M}
Let $M = \tilde{\mathbb{E}}\oplus A \oplus U$, where $\tilde{\mathbb{E}} = \begin{bsmallmatrix}
    \mathbb{E} & \mathcal{E} & \mathcal{E} & \cdots\\
    \mathcal{E} & \mathcal{E} & \mathcal{E} & \cdots\\
    \mathcal{E} & \mathcal{E} & \mathcal{E} & \cdots\\
    \svdots & \svdots & \svdots & \sddots
\end{bsmallmatrix}$ and $\mathbb{E}$ is the $n_x\times n_x$ matrix in which all elements are $e$'s.
%Then, $M^*$ is the largest matrix such that \eqref{eq:EFR_infinite} and \eqref{eq:upper_infinite} imply
    Any vector $x$ satisfying \eqref{eq:EFR_infinite} and \eqref{eq:upper_infinite} satisfies also \eqref{eq:dynamics_M}.
    Moreover, for all $\tilde{M}$ such that $\tilde{M}_{ij}> (M^*)_{ij}$ for some $i,j$, there exist $u,w$ such that $x$ satisfies \eqref{eq:EFR_infinite} and \eqref{eq:upper_infinite} but $x\neq x_0\oplus \tilde{M}x\oplus Bu\oplus Rw$.
\end{prop}
\begin{pf}
    Note that \eqref{eq:x0} implies $x(0) \geq \mathbb{E}x(0)$; therefore $x\geq \tilde{\mathbb{E}}x$.
    %Similarly, the inequality in \eqref{eq:upper_infinite} is equivalent to $x = Ux \oplus x= Ux \oplus \begin{bsmallmatrix}
    %\mathbb{E} & \mathcal{E}\\\mathcal{E}&\mathcal{E}
%\end{bsmallmatrix}x\oplus x$.
    Moreover, $x\geq Ax$ is a consequence of \eqref{eq:EFR_infinite}.
    By incorporating all inequalities, including \eqref{eq:upper_infinite}, we get $x \geq \tilde{\mathbb{E}}x\oplus Ax\oplus Ux = Mx$; equivalently, we can write $x\geq M^*x$, or $x = M^*x \oplus x$.
    Finally, by substituting \eqref{eq:EFR_infinite} into this expression, we obtain $x = M^*x \oplus x_0 \oplus Ax\oplus Bu\oplus Rw$, which coincides with \eqref{eq:dynamics_M} since $M^*\geq A$.

    To prove the second part of the proposition, suppose that $\tilde{M}_{ij}>(M^*)_{ij}$ for some $i,j$.
    Because of \ref{en:assumptionuncontrollable}, we can always choose $u$ and $w$ so that $x = M^*x_0\oplus (M^*)_{\cdot,j} \lambda$, where $(M^*)_{\cdot,j}$ denotes the $j$-th column of $M^*$ and $\lambda\in\R$ is such that $\tilde{M}_{ij}\lambda>(M^*x_0)_i$ and $(M^*)_{ij}\lambda\leq (M^*x_0)_i$.
    It can be shown that such a choice always leads to a consistent trajectory.
    Moreover, with this choice, we have
    \begin{align*}
    x_i &\!=\! (M^*x_0)_i\oplus (M^*)_{ij}\lambda \!= \!(M^*x_0)_{i} < \tilde{M}_{ij} \lambda  \\
    &\!=\! \tilde{M}_{ij} (M^*)_{jj} \lambda \!\leq \! (\tilde{M}x)_i \!\leq \!(x_0 \oplus \tilde{M}x\oplus B u \oplus Rw)_i.\mathqedhere
        %(\tilde{M}x)_i \geq \tilde{M}_{ij}x_j \geq \tilde{M}_{ij} (M^*)_{jj} \lambda = \tilde{M}_{ij} \lambda > (M^*x_0)_i.
    \end{align*}
    %as we wanted.
    %On the one hand, we have 
    %\begin{align*}
    %    [\tilde{M}x]_i &=[\tilde{M}(M^*x_0)]_i = \bigoplus_h \tilde{M}_{ih} (M^*)_{h1} \geq \tilde{M}_{ij} (M^*)_{j1} \\
    %    &> (M^*)_{ij}(M^*)_{j1};
    %\end{align*}
    %on the other, 
    %\qedhere
    %Finally, we can transform the equation in explicit form using \ref{pr:star}, obtaining \eqref{eq:dynamics_M}.
    %As seen, $x$ must satisfy \eqref{eq:Mbound} or, equivalently, $x \geq M^* x$.
    %This immediately implies: for all $k\geq 1$,
    %\begin{equation*}
    %x(k) \geq M_k x(k-1) \ \Leftrightarrow \ x(k) = M_k x(k-1) \oplus x(k).
    %\end{equation*}
    %Incorporating the latter inequality into \eqref{eq:EFR_explicit}, we get: for all $k\geq 1$,
    %\begin{equation*}
    %    x(k) = (A\oplus M_k)x(k-1)\oplus Bu(k)\oplus Rw(k).
    %\end{equation*}
    %Thus, to conclude the proof, we only need to show that $M_k\geq A$ for all $k\geq 1$.
    %Indeed, since $M^*\geq M$, we have $M_k\geq P\geq A$.\qedhere
\end{pf}

%Proposition \ref{pr:dynamics_M} gives us a necessary but not sufficient condition for consistency.
%Indeed, while in any consistent trajectory $x$ must satisfy \eqref{eq:dynamics_M}, a vector $x$ following \eqref{eq:dynamics_M} may still violate the earliest firing rule, i.e., \eqref{eq:EFR_infinite}. %too large values of $u$ and $w$ may result in upper-bound constraint violations.
%Nevertheless, from the point of view of an external observer reconstructing a consistent trajectory $x$ of a P-TEG from $u$ and $y$, \eqref{eq:dynamics_M} represents a finer description of the dynamics compared to \eqref{eq:EFR_infinite} alone, as it incorporates the information of the upper-bound constraints.
%%Compared to \eqref{eq:EFR_explicit}, the usefulness of \eqref{eq:dynamics_M} is that it provides a better estimate of the dater $x$ because it incorporates information of the upper-bound constraints of the P-TEG.
%This fact will be exploited in the development of the observer in the next section.

\begin{rem}
Because of the earliest firing rule, $x$ must be the least solution of \eqref{eq:dynamics_M} and therefore, from Theorem~\ref{th:star},
\begin{equation}\label{eq:EFR_M}
    x = M^*x_0 \oplus M^*Bu \oplus M^*Rw.
\end{equation}
%\qedhere
\end{rem}

We recall from \cite{zorzenon2025modeling} that it is possible to compute any entry $(i,j)$ of $M^*$ in a number of operations that is polynomial in $n_x$, $i$, and $j$.

\section{Observer design}\label{se:observer}

\subsection{Problem description}\label{su:problem}

Consider a P-TEG evolving with a consistent trajectory $x$, described by \eqref{eq:EFR_M} and \eqref{eq:output_infinite} (thus, it is assumed that $u$ and $w$ are such that $x$ satisfies \eqref{eq:EFR_infinite} and \eqref{eq:upper_infinite}).
We want to construct an algorithm that, at any given time $t$, provides the best estimate $\hat{x}^t_i$ of the element $x_i$ of the dater function, given only the past behavior $u^t$ and $y^t$ of $u$ and $y$ and the structure of the P-TEG (i.e., matrices $M$, $B$, $R$, and $C$). 
The estimate $\hat{x}^t_i$ to be constructed must be as large as possible, based on the available information, but not larger than $x_i$.

With this aim, we proceed in two steps.
In step 1, we consider the fastest dynamics
\begin{align}\label{eq:EFR_observer}
    %\hat{x}(0) &= e_{n_x},\label{eq:x0_observer}\\ 
    %\hat{x}(k) &= M_k \hat{x}(k-1)\oplus Bu(k) \oplus L_k C(x(k) \oplus \hat{x}(k)),\label{eq:EFR_observer}
    %\hat{x} &= \hat{x}_0 \oplus M \hat{x} \oplus B u \oplus L(y\oplus \hat{y}),
    \hat{x} &= (M\oplus LC)^*(x_0 \oplus B u \oplus Ly)
\end{align}
of a Luenberger-type observer
\begin{align*}
    \hat{x} &= x_0 \oplus M \hat{x} \oplus B u \oplus L(y\oplus \hat{y}),\\
    \hat{y} &= C \hat{x},
\end{align*}
where $L\in\Rmax^{\N_0\times \N}$ is a matrix to be designed called the \emph{observer matrix}.
We will define matrix $L_{opt}$ to be the\footnote{We will show that there is exactly one matrix with this property.} (optimal) observer matrix, for which, for all $u$ and $w$, $\hat{x}$ is the greatest vector of the form \eqref{eq:EFR_observer} less than or equal to $x$.
%In other words, we search the latest possible time $\hat{x}_i$ (in accordance with the available information and the model \eqref{eq:EFR_observer}) of the $k$-th firing of the internal transition $x_j$, where $k$ and $j$ are computed from \eqref{eq:k_and_j}, not later than the actual firing time $x_i$. %, satisfying \eqref{eq:EFR_observer}. %; thus, $\hat{x}_i(k)$ represents an estimation of of $x_i(k)$.
Note that, because operations $\oplus$ and $\otimes$ are order preserving, this objective can be stated as the one of finding the greatest observer matrix $L_{opt}$ such that, for all $u$ and $w$, $\hat{x}\leq x$; formally,
\begin{equation}\label{eq:Lopt}
L_{opt} \coloneqq \bigoplus \left\{  L\in\Rmax^{\N_0\times\N}\left| \begin{array}{l}
(\forall u,w\in\R^{\N})\ \hat{x}\leq x %\ \wedge\ \eqref{eq:EFR_M} \\
     %\wedge\ \eqref{eq:EFR_observer}\ \wedge \ L=\Pr(L) 
\end{array}  \right. \right\}. 
\end{equation}

Step 1 of our procedure will result in a mathematical formula for $L_{opt}$, which we can then use for the actual online state estimation phase (step 2).
During this phase, at any given time, only partial elements of vectors $u$ and $y$ are known: those corresponding to events that have occurred in the past.
%Let $I^t_u,I^t_y\subseteq \N$ be respectively the set of all indices $i_u$ and $i_y$ such that $u_{i_u}$ and $y_{i_y}$ are known at a given time $t\in\R_{\geq 0}$.
For any $i\in\N$, we can then obtain the best estimate $\hat{x}^t_i$ of $x_i$ at time $t$ by substituting into \eqref{eq:EFR_observer} values $L=L_{opt}$, $u = u^t$, and $y = y^t$, where
\begin{equation}\label{eq:ut_yt}
    u^t_i = 
    \begin{dcases}
        u_i & \mbox{if $u_i\leq t$,}\\
        \epsilon & \mbox{else,}
    \end{dcases}\quad
    y^t_i = 
    \begin{dcases}
        y_i & \mbox{if $y_i\leq t$,}\\
        \epsilon & \mbox{else.}
    \end{dcases}
\end{equation}
If, at any finite time $t$, vectors $u^t$ and $y^t$ contain only finitely many elements different from $\epsilon$, the calculation of estimate $\hat{x}^t_i$ can be performed in finitely many operations, despite the matrices involved being infinite.
Note that the violation of this condition on $u^t$ and $y^t$ would imply the occurrence of infinitely many events in a finite time, which is not realistic in real systems modeled by  P-TEGs.

The rest of this section is organized as follows.
Sections \ref{su:splitting} and \ref{su:solving} concern step 1 of our procedure, and provide a formula for $L_{opt}$.
Section \ref{su:online} illustrates the proposed state estimation algorithm to the P-TEG in Figure~\ref{fi:P-TEG_example}.

\subsection{Splitting the inequality}\label{su:splitting}

In order to compute $L_{opt}$, we start by substituting \eqref{eq:output_infinite} and \eqref{eq:EFR_M} into \eqref{eq:EFR_observer}.
After performing simplifications analogous to those in \cite[Eq. (22)-(24)]{hardouin2010observer}, we obtain
\[
\hat{x} = M^*(LCM^*)^*x_0 \oplus M^*(LCM^*)^*Bu \oplus M^*(LCM^*)^+ R w.
\]

The following proposition allows us to split the expression "for all $u$ and $w$, $\hat{x}\leq x$" from \eqref{eq:Lopt} into a system of several inequalities independent of $u$ and $w$.

\begin{prop}
    Let $K\in\N$, for all $k\in\dint{1,K}$, $I,J_k\subseteq\Z$, $A_k,B_k\in\Rbar^{I\times J_k}$ and $a,b\in\Rmax^I$.
    %For all $k\in\dint{1,K}$, let $A_k,B_k\in\Rbar^{I\times J_k}$ and $a,b\in\Rmax^I$, where $K\in\N$.
    The following statements are equivalent:
    \begin{itemize}
        \item for all $v_k\in\Rbar^{J_k}$, 
            \begin{equation}\label{eq:split_inequality}
            a \oplus \bigoplus_{k=1}^K A_k v_k \leq b \oplus \bigoplus_{k=1}^K B_k v_k,
            \end{equation}
        \item for all $k\in\dint{1,K}$, $A_k\leq B_k$ and $a\leq b$.
    \end{itemize}
\end{prop}
\begin{pf}
Direction "$\Leftarrow$" of the proof is immediate by the fact that $\oplus$ and $\otimes$ are order preserving.
As for direction "$\Rightarrow$", it can be proven by contrapositive.
Suppose that $(A_{\bar{k}})_{\bar{i}\bar{j}} > (B_{\bar{k}})_{\bar{i}\bar{j}}$ for some $\bar{i},\bar{j},\bar{k}$.
Then, by direct substitution, it can be verified that, by choosing $(v_k)_i = \epsilon$ for all $(i,k)\neq (\bar{i},\bar{k})$ and $(v_{\bar{k}})_{\bar{i}} = [(A_{\bar{k}})_{\bar{i}\bar{j}}\ldiv(1\otimes (a_{\bar{i}}\oplus b_{\bar{i}}))]\oplus e$, \eqref{eq:split_inequality} is violated.
If instead $a_{\bar{i}}>b_{\bar{i}}$ for some $\bar{i}$ then, by taking $(v_k)_i=\epsilon$ for all $i,k$, \eqref{eq:split_inequality} does not hold.
\qedhere
\end{pf}

%Because of the latter proposition, we can split inequality $\hat{x}(k)\leq x(k)$ into a system of inequalities, each of which depends only on  
In order to study condition "for all $u$ and $w$, $\hat{x}(k)\leq x(k)$", the latter proposition allows to separate the terms in $\hat{x}$ and $x$ that depend only on $u$, those that depend only on $w$, and those constant with respect to $u$ and $w$.
Our problem can then be restated as the one aimed at finding the largest matrix $L$ such that
\begin{align}
    \label{eq:ineq_x0} M^*(LCM^*)^*x_0 &\leq M^*x_0,\\
    \label{eq:ineq_u} M^*(LCM^*)^*B &\leq M^*B,\\
    \label{eq:ineq_w} M^*(LCM^*)^+R & \leq M^*R.
\end{align}

%In the following, we will sketch the solution of the inequality in \eqref{eq:chosen_inequality}.
%Due to their similarities, the other inequalities can be solved analogously.

\subsection{Solving the inequalities}\label{su:solving}

For solving the system of inequalities \eqref{eq:ineq_x0}-\eqref{eq:ineq_w}, we recall a result from \cite[Lemmas 1, 2, and Proposition 2]{hardouin2010observer}.
Although this result was proven in the context of the $\gamma$-transform, the proof remains identical in the domain of infinite matrices because both settings are complete idempotent semirings.

\begin{prop}
    A matrix $L$ satisfies \eqref{eq:ineq_x0}-\eqref{eq:ineq_w} if and only if $L\leq L_{opt} = J_{opt}\otimes C^T$, where
    \begin{equation}\label{eq:Jopt}
        \begin{array}{rcl}
            J_{opt} &\coloneqq& (M^*x_0)^\times \wedge (M^*B)^\times \wedge (M^*R)^\times .
            %            L_{opt} &=& (M^*x_0)\rdiv(CM^*x_0)\wedge (M^*B)\rdiv(CM^*B)\\
%                    & \wedge& (M^*R)\rdiv(CM^*R)
        \end{array}
    \end{equation}
    %where $X^\times \coloneqq X\rdiv X$.\qedhere
\end{prop}

The following proposition simplifies the expression of $L_{opt}$ by leveraging \ref{en:assumptionuncontrollable}.

\begin{prop}
    $J_{opt} = M^*$.
\end{prop}
\begin{pf}
Since, for all matrices $X$, 
\[
    (M^*X)^\times\!=\!(M^*X)\rdiv(M^*X) \!\!\eqtop{\eqref{eq:a*a*}} \!\!(M^*M^*X)\rdiv(M^*X)\!\!\geqtop{\eqref{eq:(x/a)a}} \!\!M^*,
\]
we have $J_{opt}\geq M^*$.
In the following, we show that $J_{opt}\leq M^*$, concluding the proof.
Given a matrix $X\in\Rbar^{\N_0\times \N_0}$, let us denote by $X^{ij}\in\Rbar^{n_x\times n_x}$ the $n_x\times n_x$ block of matrix $X$ in position $(i,j)$, i.e., $
%\[
    X = 
\begin{bsmallmatrix}
    X^{00} & X^{01} & X^{02} & \cdots\\
    X^{10} & X^{11} & X^{12} & \cdots\\
    X^{20} & X^{21} & X^{22} & \cdots\\
    \svdots & \svdots & \svdots & \sddots
\end{bsmallmatrix}$.
%\]
%Let $M^{ij},J_{opt}^{ij}\in\Rmax^{n_x\times n_x}$ indicate the $n_x\times n_x$ block of matrices $M^*$ and $J_{opt}$, respectively, in position $(i,j)$, i.e.,
%\[
%M^* = 
%\begin{bsmallmatrix}
%    M^{00} & M^{01} & M^{02} & \cdots\\
%    M^{10} & M^{11} & M^{12} & \cdots\\
%    M^{20} & M^{21} & M^{22} & \cdots\\
%    \svdots & \svdots & \svdots & \sddots
%\end{bsmallmatrix},
%J_{opt} = 
%\begin{bsmallmatrix}
%    J_{opt}^{00} & J_{opt}^{01} & J_{opt}^{02} & \cdots\\
%    J_{opt}^{10} & J_{opt}^{11} & J_{opt}^{12} & \cdots\\
%    J_{opt}^{20} & J_{opt}^{21} & J_{opt}^{22} & \cdots\\
%    \svdots & \svdots & \svdots & \sddots
%\end{bsmallmatrix}.
%\]
Through direct computations, it is immediate to obtain: for all $i\in\N_0$, $j\in\N$,
\begin{align*}
    [(M^*R)^\times]^{ij} &= \bigwedge_{k\in\N} ((M^*)^{ik}R_0)\rdiv((M^*)^{jk}R_0)\\
                             &\eqtop{\ref{en:assumptionuncontrollable}} \bigwedge_{k\in\N} (M^*)^{ik}\rdiv (M^*)^{jk} 
                             \leq\! (M^*)^{ij}\rdiv (M^*)^{jj} \\ &\!\!\eqtop{\eqref{eq:a**}}\!\! (M^*)^{ij}\rdiv ((M^*)^{jj})^*
                             = (M^*)^{ij}\rdiv ((M^*)^{jj}\oplus E)\\
                             &\!\!\eqtop{\eqref{eq:x/a+b}}\!\!(M^*)^{ij}\rdiv (M^*)^{jj} \wedge (M^*)^{ij} \leq (M^*)^{ij}. 
\end{align*}
Moreover, for all $i\in\N_0$,
\begin{align*}
    [(M^*x_0)^\times]^{i0} &= ((M^*)^{i0} e_{n_x})\rdiv ((M^*)^{00} e_{n_x})\\
                                  &= ((M^*)^{i0} e_{n_x})\rdiv (\mathbb{E} e_{n_x})
                                  =((M^*)^{i0} e_{n_x})\rdiv  e_{n_x} \\&=  (M^*)^{i0} e_{n_x} e_{n_x}^T 
                                  = (M^*)^{i0} (M^*)^{00} = (M^*)^{i0},
\end{align*}
where the latter equivalence is proven in \cite[Proposition 6.1]{zorzenon2025modeling}.
Since $J_{opt}\leq (M^*R)^\times$ and $J_{opt}\leq (M^*x_0)^\times$, we have proven that, for all $i,j\in\N_0$, $J_{opt}^{ij}\leq (M^*)^{ij}$; therefore, $J_{opt}\leq M^*$.\qedhere
\end{pf}

By substituting $L_{opt}$ into \eqref{eq:EFR_observer}, we obtain that the largest $\hat{x}$ satisfying $\hat{x}\leq x$ for all $u,w$ is given by
\begin{align}
    \hat{x} &= (M\oplus M^* \underbrace{C^TC}_{\leq E})^*(x_0\oplus Bu \oplus M^*C^Ty)\notag\\[-.2cm]
            &\eqtop{\eqref{eq:a*a*}} M^*x_0 \oplus M^*Bu \oplus M^*C^Ty.\label{eq:observer_solution}
\end{align}

\subsection{Online estimation}\label{su:online}

We now illustrate the online state estimation procedure explained in Section \ref{su:problem} for the P-TEG example of Figure~\ref{fi:P-TEG_example}. 
The results presented in the following can be replicated by running the MATLAB function \texttt{observerPTEGs.m} from the library TW\textsubscript{max}, which can be found at: \url{https://git.tu-berlin.de/davide_zorzenon/twmax}.

We will suppose that: the first four firing times of transitions $w_1,w_2,w_3$ are as in \eqref{eq:wexample}, the first four firings of $u_1$ all occur at time $0$, while all subsequent firings of $u_1$ occur after time $32$.
According to this, the first four firings of the internal transitions occur as indicated in \eqref{eq:xexample}.
Since $y_1(1) = 4$, at time $t=0$ only the firing times $u_1(1),\dots,u_1(4)$ are known.
Following \eqref{eq:ut_yt}, this implies that $u^0_i = 0$ for $i\in\dint{1,4}$, $u^0_i = \epsilon$ for all $i>4$, and $y^0_i = \epsilon$ for all $i\in\N$.
We can then compute any element $\hat{x}^0_i$ of estimate $\hat{x}^0$ by substituting values $u=u^0$ and $y=y^0$ in \eqref{eq:observer_solution}.
Because of the sparsity of $x_0$, $u^0$, and $y^0$, this only requires the knowledge of values $(M^*)_{ij}$ for $j\in\dint{1,n_x}=\dint{1,3}$ to compute $(M^*x_0)_i$ and for $j\in\dint{n_x+1,4\cdot n_x} = \dint{4,12}$ to compute $(M^*Bu)_i$ (the number $4$ appears in $4\cdot n_x$ because there are exactly $4$ real entries in $u^0$); the term $(M^*C^Ty)$ does not need to be computed because $y$ is $\epsilon$ everywhere.

We can generalize these facts to conclude that, since each element $(M^*)_{ij}$ can be obtained in time polynomial with respect to $i$, $j$, and $n_x$, for all times $t$, $\hat{x}_i^t$ can be computed in polynomial time with respect to $i$, the number of real entries in $u^t$ and $y^t$, and $n_x$.

Returning to our example, the resulting estimates of $x(1),\ldots,x(4)$ at time $t=0$ are (estimates matching actual firing times from \eqref{eq:xexample} are \underline{underlined})
\[
    \hat{x}^0(1) = 
    \begin{bmatrix}
        \underline{2}\\3\\3
    \end{bmatrix},\,
    \hat{x}^0(2) = 
    \begin{bmatrix}
        4\\5\\5
    \end{bmatrix},\,
    \hat{x}^0(3) = 
    \begin{bmatrix}
        6\\7\\7
    \end{bmatrix},\,
    \hat{x}^0(4) = 
    \begin{bmatrix}
        8\\9\\9
    \end{bmatrix}.
\]
%which correspond to the fastest consistent trajectory of the P-TEG.
At times $t=4$, $t=10$, $t=20$, and $t=31$, due to the new information coming from the firings of transition $y_1$, the estimates are updated to the following values (new values compared with previous estimates are in \textbf{bold}):
\[
    \hat{x}^4(1) = 
    \begin{bmatrix}
        \underline{2}\\\underline{\textbf{4}}\\\underline{\textbf{4}}
    \end{bmatrix},\,
    \hat{x}^4(2) = 
    \begin{bmatrix}
        4\\5\\5
    \end{bmatrix},\,
    \hat{x}^4(3) = 
    \begin{bmatrix}
        6\\7\\7
    \end{bmatrix},\,
    \hat{x}^4(4) = 
    \begin{bmatrix}
        8\\9\\9
    \end{bmatrix};
\]
%at time $t=10$ to
\[
    \hat{x}^{10}(1) \!= \!
    \begin{bmatrix}
        \underline2\\\underline4\\\underline4
    \end{bmatrix}\!,
    \hat{x}^{10}(2) \!= \!
    \begin{bmatrix}
        \textbf{5}\\\underline{\textbf{10}}\\\underline{\textbf{10}}
    \end{bmatrix}\!,
    \hat{x}^{10}(3) \!=\! 
    \begin{bmatrix}
        \textbf{7}\\\textbf{10}\\\textbf{10}
    \end{bmatrix}\!,
    \hat{x}^{10}(4) \!= \!
    \begin{bmatrix}
        \textbf{9}\\\textbf{10}\\\textbf{10}
    \end{bmatrix}\!;
\]
%at time $t=20$ to
\[
    \hat{x}^{20}(1) \!= \!
    \begin{bmatrix}
        \underline{2}\\\underline4\\\underline4
    \end{bmatrix}\!,
    \hat{x}^{20}(2) \!= \!
    \begin{bmatrix}
        5\\\underline{10}\\\underline{10}
    \end{bmatrix}\!,
    \hat{x}^{20}(3) \!=\! 
    \begin{bmatrix}
        \textbf{15}\\\underline{\textbf{20}}\\\underline{\textbf{20}}
    \end{bmatrix}\!,
    \hat{x}^{20}(4) \!= \!
    \begin{bmatrix}
        \textbf{17}\\\textbf{20}\\\textbf{20}
    \end{bmatrix}\!;
\]
%and, finally, at time $t=31$ to
\[
    \hat{x}^{31}(1) \!= \!
    \begin{bmatrix}
        \underline2\\\underline4\\\underline4
    \end{bmatrix}\!,
    \hat{x}^{31}(2) \!= \!
    \begin{bmatrix}
        \textbf{6}\\\underline{10}\\\underline{10}
    \end{bmatrix}\!,
    \hat{x}^{31}(3) \!=\! 
    \begin{bmatrix}
        \underline{\textbf{16}}\\\underline{20}\\\underline{20}
    \end{bmatrix}\!,
    \hat{x}^{31}(4) \!= \!
    \begin{bmatrix}
        \underline{\textbf{26}}\\\underline{\textbf{31}}\\\underline{\textbf{31}}
    \end{bmatrix}\!.
\]
Thus, eventually the observer converges to the actual firing times of internal transitions with the only exception of the second firing time of transition $x_1$, which is underestimated by $1$ time unit.
As an example, it is worth commenting on how the observer realizes, at time $31$, that the estimate $\hat{x}^{20}_1(2)=5$ needs to be adjusted: it is impossible that transition $x_1$ fires the second time at time $5$ and the third time at time $16$, as this would cause a token death in the place with time-window constraint $[2,10]$.
The earliest firing time not resulting in token deaths is thus $\hat{x}^{31}_1(2)=6$.

Let us compare our algorithm to an observer that does not exploit the information coming from upper-bound constraints (like the one from~\cite{hardouin2010observer}).
Such an observer can be constructed following the same procedure described in this paper, with the only difference that matrix $M$ is defined by $M = A$ instead of $M = \tilde{\mathbb{E}} \oplus A \oplus U$.
The estimates of $x(k)$ that this observer achieves at time $t = 31$, denoted by $\hat{x}^A(k)$, are shown below:
\[
    \hat{x}^A(1) \!= \!
    \begin{bmatrix}
        \underline{2}\\2\\\underline{4}
    \end{bmatrix}\!,
    \hat{x}^A(2) \!= \!
    \begin{bmatrix}
        4\\4\\\underline{10}
    \end{bmatrix}\!,
    \hat{x}^A(3) \!=\! 
    \begin{bmatrix}
        6\\6\\\underline{20}
    \end{bmatrix}\!,
    \hat{x}^A(4) \!= \!
    \begin{bmatrix}
        8\\8\\\underline{31}
    \end{bmatrix}\!.
\]
Except for $x_3(k)$ (which is measured), all estimations are worse than or equal to the ones obtained by our observer at time $t=0$.
%A visual comparison between the two methods is provided in Figure~\ref{fi:estimation_error}, where estimation error $x_1(k)-\hat{x}^t_1(k)$ for the $k$th firing of transition $x_1$, with $k\in\dint{2,4}$ and $t\in\{0,4,10,20,31,A\}$, is represented.
%
%\begin{figure}
%    \centering
%    \resizebox{.8\linewidth}{!}{
%    \input{figures/estimation_error}
%    }
%    \caption{Comparison of the estimation errors.}
%    \label{fi:estimation_error}
%\end{figure}

%\section{Conclusions}

%\appendix
%\input{appendix.tex}

\bibliography{references}

@book{baccelli1992synchronization,
  title={Synchronization and linearity: an algebra for discrete event systems},
  author={Baccelli, Fran{\c{c}}ois and Cohen, Guy and Olsder, Geert Jan and Quadrat, Jean-Pierre},
  year={1992},
  publisher={John Wiley \& Sons Ltd}
}

@inproceedings{khansa1996p,
  title={P-time {P}etri nets for manufacturing systems},
  author={Khansa, Wael and Denat, Jean-Paul and Collart-Dutilleul, Simon},
  booktitle={International Workshop on Discrete Event Systems, WODES},
  volume={96},
  pages={94--102},
  year={1996}
}

@article{hardouin2018control,
  title={{Control and state estimation for max-plus linear systems}},
  author={Hardouin, Laurent and Cottenceau, Bertrand and Shang, Ying and Raisch, J{\"o}rg},
  journal={Foundations and Trends{\textregistered} in Systems and Control},
  volume={6},
  number={1},
  pages={1--116},
  year={2018},
  publisher={Now Publishers, Inc.}
}

@ARTICLE{hardouin2010observer,
  author={Hardouin, Laurent and Maia, Carlos Andrey and Cottenceau, Bertrand and Lhommeau, Mehdi},
  journal={IEEE Transactions on Automatic Control}, 
  title={Observer Design for  $(\max, +)$ Linear Systems}, 
  year={2010},
  volume={55},
  number={2},
  pages={538-543},
  doi={10.1109/TAC.2009.2037477}
}

@mastersthesis{tirpak2024development,
  TITLE = {Development and implementation of an optimal state observer for {P}-time event graphs},
  AUTHOR = {Tirp\'{a}k, Dominik},
  SCHOOL = {Technische Universit\"{a}t Berlin},
  YEAR = {2024},
  MONTH = May,
  TYPE = {Master's thesis}
}

@phdthesis{zorzenon2025modeling,
    title={Modeling, Analysis, and Control of Discrete-Event Systems Described by Time-Window Constraints},
    author={Zorzenon, Davide},
    school = {Technische Universit\"{a}t Berlin},
    year = {2025},
    month = {June},
    doi = {10.14279/depositonce-24794}
}

@ARTICLE{bonhomme2015marking,
  author={Bonhomme, Patrice},
  journal={IEEE Transactions on Systems, Man, and Cybernetics: Systems}, 
  title={Marking Estimation of {P}-Time {P}etri Nets With Unobservable Transitions}, 
  year={2015},
  volume={45},
  number={3},
  pages={508-518},
  doi={10.1109/TSMC.2014.2353575}
}

@article{TRUNK2020501,
title = {Observer for Weighted Timed Event Graphs},
journal = {IFAC-PapersOnLine},
volume = {53},
number = {4},
pages = {501-507},
year = {2020},
note = {15th IFAC Workshop on Discrete Event Systems WODES 2020},
issn = {2405-8963},
doi = {10.1016/j.ifacol.2021.04.046},
author = {J. Trunk and G. Schafaschek and B. Cottenceau and L. Hardouin and J. Raisch}
}

@ARTICLE{Espindola2022stochastic,
  author={Espindola-Winck, Guilherme and Hardouin, Laurent and Lhommeau, Mehdi and Santos-Mendes, Rafael},
  journal={IEEE Transactions on Automatic Control}, 
  title={Stochastic Filtering Scheme of Implicit Forms of Uncertain Max-Plus Linear Systems}, 
  year={2022},
  volume={67},
  number={8},
  pages={4370-4376},
  doi={10.1109/TAC.2022.3176841}}

@article{di2010duality,
  title={Duality between invariant spaces for max-plus linear discrete event systems},
  author={Di Loreto, Michael and Gaubert, Stephane and Katz, Ricardo D and Loiseau, Jean-Jacques},
  journal={SIAM Journal on Control and Optimization},
  volume={48},
  number={8},
  pages={5606--5628},
  year={2010},
  publisher={SIAM}
}

\end{document}